\documentclass[a4paper]{article}
\usepackage{fullpage,url,amsmath,amsfonts,amssymb,tedmath}
\newcommand{\ord}{\rm{OrderMod}}
\title{Quantum
error-correcting  codes: the  unit design strategy}
\author{Ted Hurley\footnote{National University of Ireland Galway. Ted.Hurley@NuiGalway.ie}, Donny Hurley\footnote{Institute of Technology, Sligo. hurleyd@yahoo.com }, Barry Hurley\footnote{barryj\_2001@yahoo.co.uk}}
\date{}
\begin{document}
\maketitle
\begin{abstract}\let\thefootnote\relax %% Quantum mds (maximum distance separable) codes are constructed. 
Series of maximum distance quantum error-correcting codes are developed and analysed. 
  For a given rate and given error-correction capability,  quantum error-correcting codes with these specifications are constructed. The codes are explicit with efficient decoding algorithms. For a given field maximum length quantum codes are constructed. 
\end{abstract}
\section{Introduction}

Quantum error-correcting codes have an important role in quantum computing and are used  to protect quantum information from errors due to quantum noise and 
decoherence. %% Quantum error correction is essential if one is to achieve fault-tolerant quantum computation that can deal not only with noise on stored quantum information, but also with faulty quantum gates, faulty quantum preparation, and faulty measurements.

A short introduction to quantum coding,  with  history and
bibliographical notes, is given  in Kim and Matthews
\cite{kim}. Background information on quantum codes and quantum
information theory may be found in  the book by Nielsen and Chung
\cite{nielsen}.   A survey article with emphasis on topological aspects
of quantum computing appears in Rowell and Wang \cite{notices}.    

The literature on quantum error-correcting codes is massive and expanding. 

Required background on (classical) coding theory and on basic algebra, including in particular {\em Field Theory}, 
 may be found in \cite{blahut} or in \cite{mac}.  $GF(q)$ will denote the finite field of order $q$; of necessity $q$ is a power of a prime. 
The non-zero elements of $GF(q)$ form a cyclic group of order $(q-1)$ and any generator of the group is termed a {\em primitive} element of $GF(q)$. 
A (classical) code of length $n$, dimension $k$ and distance $d$ over $GF(q)$ is denoted by $[n,k,d]_q$ or simply by $[n,k,d]$ when the field is understood or given.

  Seminal work of Calderbank, Shor and  Steane \cite{calderbank,good,steane1} provide  the relationship between  classical codes and quantum error-correcting codes. Their construction 
is now known as the {\em CSS construction}. Following Rains' works
\cite{rains} on nonbinary quantum codes, the work of Calderbank, Shor and
Steane was extended to nonbinary cases by Ashikhmin and Knill,
\cite{ash,ketner}. Gottesman \cite{gott} had previously developed the stabilizer formalism for quantum codes.  See also work of Shor and Steane in \cite{maybe2,shor1,steane}. 

An $[[n,k,d]]_q$ code is a quantum code  of length $n$, dimension $k$
and minimum distance $d$ over the field $GF(q)$; the word {\em $q-$ary
code} is sometimes used.  When the field is understood or given, the
notation  $[[n,k,d]]$ can be used, without the $q$ suffix. Use
$[[n,k,\geq d]]$ to mean a quantum code of length $n$, dimension $k$ and minimum distance at least $d$.  

A classical $[n,k,d]$ code satisfies the Singleton bound $d\leq (n-k+1)$ and a code reaching this bound is called an {\em mds (maximum distance separable) code}. A quantum $[[n,r,d]]$ code satisfies the quantum Singleton bound $2d \leq (n-r+2)$ and a quantum code attaining the bound  is called an {\em mds quantum code}.
\subsection{The quantum codes}
Here series of mds quantum codes are constructed. For a given rate and
given error-correcting capability quantum error-correcting mds codes of
this rate and capability are given. Efficient decoding algorithms are available. The
codes are constructed from dual-containing linear mds codes derived
using the constructions in \cite{hurley}. For each characteristic,
quantum mds codes are constructed over finite fields with this
characteristic and with given rate and given error-correcting capability
provided the 
characteristic does not divide the then required length. For a given
rate and required  error-correcting capability,  quantum mds codes are constructed over a
field of prime order with this rate and achieving the  required
error-correcting capability. For a given finite field, best quantum mds codes  
 are constructed over this field.

%Of necessity $q$ is a power of a prime.    

\subsection{CSS construction}
The CSS constructions in use here are specified  as follows: 

\begin{itemize} 
\item Let $\mathcal{C}$ be a classical linear code $[n,k,d]$ over $GF(q) $ containing its dual $\mathcal{C}^\perp$. The CSS construction derives a quantum (stabilizer)  $[[n,2k-n,\geq d]]$ code over $GF(q)$. 
\item Let $\mathcal{D}$ be a classical linear code over $GF(q^2)$ containing its Hermitian dual  $\mathcal{D}^{\perp_H}$. The CSS construction derives a quantum (stabilizer) code $[[n,2k-n, \geq d]]$ code over $GF(q^2)$.
\end{itemize}

For more details on CSS constructions of quantum error-correcting codes
see \cite{ash,ketner}; proofs of the above  may also be found therein. The work of \cite{ash} follows from Rains' work on nonbinary codes \cite{rains}.  
%% in \cite{ash} and \cite{ketner}.  

A code containing its dual is called a {\em dual-containing} code.

As noted in for example  \cite{ash} if the dual-containing code used for the CSS construction is an mds linear code then the quantum code obtained is a quantum mds code.   

\subsection{The dual-containing codes}
In \cite{hurley} systems of  mds linear codes, with efficient decoding algorithms,  are constructed and analysed using 
Vandermonde/Fourier matrices. These codes are defined using the unit-derived methods of \cite{hur1}. %% these also have efficient decoding
%% algorithms. 
The constructions in \cite{hurley} are  now used to define series of 
 dual-containing
codes, and Hermitian dual-containing codes,  from which mds quantum codes are constructed by the CSS
constructions. The dual-containing codes obtained have efficient
decoding algorithms by \cite{hurley} giving efficient decoding
algorithms for the quantum codes constructed. 
 
In \cite{hurley100} binary and characteristic $2$  quantum codes are
constructed from  group 
rings, with  matrix implementations, but the methods there have been
ignored.   The
paper \cite{lin} constructs quantum codes from matrix product codes.
%% For example $[[2^{2n}-1, r, 2^n-1]]$ mds quantum codes are
%% constructed and these have efficient decoding algorithms. %% (The $r$ can
%% %% be written in terms of the length and distance as the code is mds.)  

\section{The constructions} 

Let $\om$ be a primitive $n^{th}$ root of unity in a field $\F$; primitive here means that $\om^n=1$ and $\om^r \neq 1$ for $1\leq r < n$. For
such an $\om$ to exist in $\F$ it is necessary that the characteristic of $\F$ 
does not divide $n$ and in this case $n$ has an inverse in $\F$.  

The Fourier $n\ti n$ matrix, relative to $\om$,  is the $n\ti n$ matrix  

$$F_n= \begin{pmatrix} 1 &1 &1 & \ldots & 1 \\ 1 & \om & \om^2 & \ldots
       & \om^{n-1} \\ 1&\om^2 & \om^{2(2)} & \ldots & \om^{2({n-1})} \\
       \vdots & \vdots & \vdots & \ldots & \vdots \\ 1 & \om^{n-1} &
       \om^{(n-1)2} & \ldots & \om^{(n-1)(n-1)} \end{pmatrix} $$. 

The inverse of $F_n$ is obtained directly by replacing $\om$ by $\om^{n-1}$ in the above formula and dividing by $n$;  $n^{-1}$ exists in $\F$. The inverse of $F_n$ is not required explicitly here but is there in the background.   %% as a primitive $n^{th}$ root of unity exists in $\F$. 
The Fourier matrix is a type of Vandermonde matrix where the entries are roots of unity. 

The rows, in order,  of a Fourier matrix $F_n$ under consideration 
will be denoted by $\{e_0, e_1, \ldots, e_{n-1}\}$. Thus $e_i = (1,
\om^i, \om^{i(2)}, \ldots, \om^{i(n-1)})$ for the primitive $n^{th}$ root of unity $\om$ used to build $F_n$. 

We refer to \cite{hurley} for the following theorem:

\begin{theorem}{\cite{hurley}}\label{seq} Let $\mathcal{C}$ be a code generated by taking any $r$ rows of $F_n$ in arithmetic sequence with arithmetic difference $k$ satisfying $\gcd(n,k) = 1$. Then $\mathcal{C}$ is an mds (maximum distance separable) $[n,r,n-r+1]$ code.
\end{theorem} 
In particular we have:
\begin{theorem}{\cite{hurley}}\label{seq1} Let $\mathcal{C}$ be a code generated by taking $r$ consecutive rows of $F_n$. Then $\mathcal{C}$ is an mds $[n,r,n-r+1]$ code.
\end{theorem}

 Efficient decoding algorithms are in addition given in that paper
 \cite{hurley}. The {\em unit-derived strategy} as in \cite{hur2,hur1}
 is used for the constructions and analysis. The Theorems are based  on
 methods originally derived in \cite{barr}. % and  

 Define the Euclidean inner product of vectors $u = (u_1,u_2, \ldots,
 u_n), v=(v_1,v_2, \ldots, v_n) \in \F^n$ by $<u,v>_E = u_1v_1 + u_2v_2+
 \ldots u_nv_n$;  this is an element of $\F$. An Hermitian inner product
 will be defined over the field $GF(q^2)$ in Section \ref{her}. For
 the Euclidean inner product, omit the suffix $E$ so that $<u,v> =
 <u,v>_E$.  %% that is, in the Euclidean inner product the suffix $E$
 %% may be omitted.  

Let $C$ be a subspace of $\F^n$. The (Euclidean) dual of $C$ is defined by 
$C^\perp = \{u\in \F^n | <u,v> = 0, \forall v\in C\}$. % \footnote{Later on the Hermitian dual $C^\perp_H$ will be considered.}
 If $C$ has dimension $r$ then $C^\perp$ has dimension $n-r$. 

Let $\{e_0,e_1, \ldots, e_{n-1}\}$ denote the rows, in order,  of a Fourier matrix
$F_n$ over a field $\F$. Note that $\{e_0,e_1, \ldots, e_{n-1}\}$ is a basis for $\F^n$. This basis is not an orthogonal basis (luckily!) relative to
the Euclidean inner product. %% Now $<u,v>_E$ denotes the Euclidean inner
%% product which we will denote simply by $<u,v>$ without the suffix $E$.

The inner product of the rows $e_i$ of the Fourier $n\ti n$ matrix satisfy   
$<e_i,e_j>_E= e_ie_j\T = 1 + \om^{i+j} + \om^{2(i+j)}+ \ldots +
\om^{(n-1)(i+j)}$. %%  and thus $<e_i,e_j> = 0, i+j\neq 0 \mod n$ and
%% $<e_i,e_{n-i}> = n$.

Notice that $<e_i, e_{n-i}> = n$ (where $e_n$ is interpreted as
$e_0$) and that $<e_i,e_j> = 0$ when $j\neq n-i$. This enables the %easy
calculation of the dual of a code generated by the rows of a Fourier
matrix. 

\begin{proposition}\label{dual} Let $\mathcal{C}$  be the code generated by the
 distinct rows $<e_{i_1},e_{i_2}, \ldots ,
 e_{i_r}>$ of the Fourier $n\ti n$ matrix $F_n$. Then $\mathcal{C}^\perp$ is generated by
 the rows of $F_n$ excluding $\{e_{n-i_1},e_{n-i_2}, \ldots, e_{n-i_r}\}$.
\end{proposition}

\begin{proof} This follows directly since $<e_i,e_j> = 0$ when $j\neq
 n-i$.
\end{proof}

The generators for $\mathcal{C}^\perp$ derived in Proposition \ref{dual} are  a basis and  $\mathcal{C}^\perp$ has dimension $(n-r)$.

For example  suppose $\{e_0,e_1,\ldots , e_{9}\}$ are the rows of a Fourier
$10\ti 10$ matrix. \\ Let $\mathcal{C}= <e_0,e_1,e_2,e_3,e_4,e_5>$. Then
$\mathcal{C}^\perp$ is generated by the rows excluding $\{e_0, e_{9},
e_{8},e_7, e_6,e_5\}$ and thus $\mathcal{C}^\perp =
<e_1,e_2,e_3,e_4>$. Note in this example, by \cite{hurley}, that  $\mathcal{C}$
is a $[10,6,5]$ code and now we see that it contains its dual $\mathcal{C}^\perp
=<e_1,e_2,e_3,e_4>$.  By CSS construction a $[[10,2,5]]$ quantum mds code
is obtained. By taking $\mathcal{D}=<e_0,e_1, \ldots, e_7>$ get that 
$\mathcal{D}^\perp = <e_1,e_2>$ and so again $\mathcal{D}$ is a dual
containing $[10, 8,3]$ code which by CSS construction gives a
$[[10,6,3]]$ mds quantum code. If however $\mathcal{T}=
<e_0,e_1,e_2,e_3,e_4>$ then $\mathcal{T}^\perp = <e_1,e_2,e_3,e_4,e_5>$
and so %% , as expected as rank $\mathcal{T}\leq 1/2 $,
$\mathcal{T}$ is not dual-containing;  it will be noticed later that for the rows of a Fourier matrix to generate a dual-containing code it is necessary that more than half the  number of rows need to be involved.

Codes from \cite{hurley} which can be shown to be dual-containing relative to the Euclidean
inner product  are used to construct mds quantum codes by the CSS construction. 

 In section \ref{her} codes from \cite{hurley} which are dual-containing relative to a 
Hermitian inner product, over fields of form $GF(q^2)$, are used  to form 
mds quantum codes. % are deduced. 

Here are some examples with which to begin: 
\begin{enumerate}
% \item Consider $GF(2^2)$. Here the Fourier $3\ti 3$ matrix over $GF(2^2)$ exists. Take
% the first two rows of this matrix to get a $[3,2,2]$ dual-containing code 
% from which the $[[3,1,2]]$ quantum code is obtained; that the code has distance $2$ follows from Theorem \ref{seq1}.  

\item $GF(2^5)$. Here the Fourier $31\ti 31$ matrix $F_{31}$ exists over $GF(2^5)$. Let
$r> 15$ and form the code generated by the first $r$  rows of
$F_{15}$ to form the $[31, r, 31-r+1]$ code which is mds by Theorem \ref{seq1} and is 
dual-containing by Theorem \ref{greater}. This gives by the CSS construction the
$[[31, 2r-31, 31-r+1]]$ mds quantum code. For instance if $r=25$ this
gives the $[[31,19,7]]$ mds quantum code over $GF(2^5)$. 
   
\item $GF(2^8)$. The length here is that of the Reed-Solomon codes. Since $2^8-1 =
255$ there exists a Fourier $255\ti 255$ matrix over
$GF(2^8)$. Codes may be obtained by taking $r$ rows in
succession or indeed $r$ rows in arithmetic sequence with difference
$k$ satisfying $\gcd(255,k)=1$, to form an $[255,r,255-r+1]$ code which will be
dual-containing by Theorem \ref{greater} provided $r\geq 128 $. From
this, using the CSS construction,  $[[255,2r-255, 256-r]]$ quantum mds codes
are formed. For example with $r= 245$ get $[[255, 235,11 ]]$ mds quantum
code. 
\item Consider the field $GF(257)$ and note that  $257$ is prime. The field  has primitive $256^{th}$ roots of unity, and then form the Fourier $256\ti 256 $ Fourier matrix with a  one of these primitive roots. For example $3 \mod 257$ has order $256$ in $GF(257)$.
Now take the first $r$ rows of this Fourier matrix with $r> 128$ to form a $[256, r, 256-r+1]$ dual-containing code from which by the CSS construction the $[[256, 2r-256, 257-r]]$ quantum code may be formed. These examples may be compared with those in the previous example; the arithmetic here is modular arithmetic,  which is easy to implement, and one works over $\Z_{257}$. 
\end{enumerate}

\section{Required  rate and error-correcting capability}\label{required}
For a given rate $R$ and error-correcting capability it is required  
to construct a quantum code with this rate and distance greater than or equal to a given distance, that
is, it is required to  construct a quantum code of the from $[[n,k,\geq d]]$ for $R=\frac{k}{n}$ and for a given $d$. 

To construct such a code,  construct first of all a dual-containing code
which when the CSS  construction is applied will give the required
quantum code. In fact an mds quantum code with given
rate and given error-correcting  capability will be constructed. Note that the 
CSS constructions produce codes of the form $[[n,2r-n,d]]$. % a rate $R$ dual-containing code which
				 % has required distance $d$.  

A $[[n,k,d]]$ quantum code must satisfy the quantum Singleton Bound $2d\leq n-k+2$ for $k>1$ and the codes constructed reach this bound. 

A dual-containing code which is not self-dual must have rate greater than $\frac{1}{2}$. If it has rate
equal to $\frac{1}{2}$ it would be a self-dual code  and the corresponding CSS
construction would give a rate of $0$.  Assume here when constructing $[[n,k,d]]$ that $k> 0$.  

\begin{theorem}\label{greater} Let  $\{e_0,e_1, \ldots, e_{n-1}\}$ be the
 rows of a Fourier $n\ti n$ matrix $F_n$. Suppose  $2t\geq
 n-1$. Then the  code $\mathcal{C}$ generated by $\{e_0,e_1, \ldots, e_t\}$ (which has $t+1$ elements) is a dual-containing $[n,t+1, n-t]$ mds code. 
\end{theorem}
\begin{proof}
That $\mathcal{C}$ is an $[n,t+1, n-t]$ mds code follows from \cite{hurley}, Theorem \ref{seq1} above. 
The dual $\mathcal{C}^\perp$ of $\mathcal{C}$ is  generated by the rows of $F_n$ excluding $\{e_0, e_{n-1}, e_{n-2}, \ldots, e_{n-t}\}$ by Proposition \ref{dual}. Thus $\mathcal{C}^\perp = <e_1, e_2, \ldots, e_{n-t-1}>$. Now $n-t-1 < t$ as $2t \geq n-1$ and thus $\mathcal{C}^\perp \subset \mathcal{C}$.     
\end{proof}

Thus for $R=\frac{r}{n} > \frac{1}{2}$  an $[n,r,d]$ mds dual-containing code may be
built from a Fourier $n\ti n $ matrix by the methods of
\cite{hurley}. These codes have efficient decoding algorithms by methods
of \cite{hurley}.  

For example let $n=12$ and consider the Fourier $12 \ti 12 $ matrix with
rows $\{e_0,e_1,\ldots, e_{11}\}$. Then the code generated by $\{e_0,e_1,
e_2,e_3,e_4,e_5,e_6\}$ is a $[12,7,6]$ dual-containing code and also
$\{e_0, e_1, \ldots, e_8\}$ is a $[12,9,4]$ dual-containing code;
however the  code generated by $\{e_0, e_1,e_2,e_3,e_4,e_5\}$ is not
dual-containing. Another way to obtain a dual-containing code is by choosing $\{e_6,e_7,e_8,e_9,e_{1},e_{11},e_0\}$; these are in order of the Fourier matrix.  Note that $\{e_0, e_6\}$ must be included, since $<e_0,e_0> \neq 0$ and $<e_6,e_6> \neq 0$.  

%Rewrite the next paragraph. 
{\em Strictly} greater  than half the elements must be taken in cases here; self-dual codes cannot be obtained by the method. %% The elements start at the first one and are taken in order but as noted the last $t$ elements with $2t\geq n-1$ where the last element is taken as $e_n=e_0$. %f
%or generating the 

For given $n$ does there exist a Fourier $n\ti n$ matrix over some
finite field $GF(q)$? Does there exist a Fourier $n\ti n$ matrix over a
finite field of given characteristic $p$? 
 
The following is taken from \cite{hurley}. For a Fourier $n\ti n$ 
matrix to exist over $GF(p^t)$ for a prime $p$ it is necessary that $p$ does not divide $n$. Let $p$ be a prime not dividing $n$. Then by Euler's Theorem, $p^{\phi(n)}\equiv 1 \mod n$. Thus $p^{\phi(n)} -1 = nq$ for some integer $q$. Let $\F= GF(p^{\phi(n)})$. Then $\F$ has a primitive element $\be$ of order $p^{\phi(n)}-1$. 
Now let $\om = \be^q$. Then $\om$ is an element of order $n$ in $\F$. Hence a Fourier $n\ti n$ matrix may be constructed over $\F$ which is a field of characteristic $p$. 

Although $GF(p^{\phi(n)})$ works,  in many cases smaller fields of
characteristic $p$ may be obtained over which a Fourier $n\ti n$ matrix
exists; it depends on the order of $p \mod n$ which exists when $p$ does not divide $n$.

Suppose now a rate $R$ and  a distance $d$ are given and it is
required to build a $[[n,k,\geq d]]$ quantum code with $R=\frac{k}{n}$. We can assume that $n+k$ is even.  Note when acquiring a $[[n,2r-n,\geq d]]$ code from a dual-containing code $[n,r,d]$ code by CSS construction that $n+(2r-n)$ (the length plus the dimension) is always even.
 
Let
$2r=n+k$. Then $\frac{r}{n} = \frac{1}{2}+\frac{k}{2n} $ and thus
$\frac{1}{2} < \frac{r}{n} <1 $. Thus 
build a $[n,r,d]$ dual-containing code achieving the distance  
$d=n-r+1$. The rate of this code must be $\frac{r}{n}$ and $d= n-r+1 =
n-\frac{n+k}{2}+1 = n- \frac{n+nR}{2}+1 = n(1-\frac{1+R}{2}) +1 $. Thus
it is required that $n=(d-1)/(1-\frac{1+R}{2})=
\frac{2(d-1)}{1-R}$. Suppose now $(1-R)= \frac{p}{q}$ as a reduced
fraction. Since $n$ is an integer, it is required further that $p/(2(d-1))$. Then
consider distances $d=d_0,d+1=d_1, \ldots, $ until $p/(2(d_i-1))$; 
this gives a value for $n$ and hence a value for $r$ and $k$. 
Now form the $n\ti
n$ Fourier matrix and choose $r$ rows in succession (or in suitable arithmetic sequence as per Theorem \ref{seq}). This gives a
$[n,r,n-r+1]$ dual-containing mds code with $n-r+1\geq d$. By the CSS construction a
$[[n,2r-n, n-r+1]]$ quantum code is constructed. Now $2r-n = n+k - n =k$ and thus a quantum mds code of required rate and   required error-correcting capability has been constructed. 

Note that if $R=\frac{t-1}{t}$ then $1-R = \frac{1}{t}$ and automatically
the formula $n=\frac{2(d-1)}{1-R}$ gives $n$ as a positive integer. 
%problem with n being bad/not suitable. 

 \subsection{Examples of required rate and error-correcting capability}

For positive integers $t,v$, use \ord$(t,v)$ to mean the order of $t \mod v$ when such exists. 

\begin{itemize} 
\item Suppose a rate $\frac{3}{4}$ quantum code is required
 which can correct one error, that is a $[[n,k,\geq 3]]$ quantum code with
 $\frac{3}{4} = \frac{k}{n}$ is required. Let $2r=n+k$. Now build a
 $[n,r,n-r+1]$ dual-containing code. Require $n-r+1 = 3 $ and as above
      get $n=\frac{2(d-1)}{1-R}= 8(d-1)=16$ when $d=3$. Thus  $n=16, k=12, r=14$. It is thus
 required to build a $[16,14,3]$ dual-containing code from which the
 $[[16, 12, 3]]$ quantum code may be derived by the CSS construction. 

A $16\ti 16$ Fourier matrix is required. Now $3$ is the first prime
 with $\gcd(3,16)=1$ and $\ord(3,16)= 4$. Thus a Fourier $16\ti 16 $
 matrix exists over $GF(3^4)$. Take the first $14$ rows of this to get a
 $[16,14,3]$ dual containing code from which the $[[16,12,3]]$ quantum
 code is obtained. 

Now \ord$(7,16)=2$ so a Fourier $16\ti 16$ matrix exists over
 $GF(7^2)$ and this may be better for constructing the $[[16,14,3]]$
 quantum code. But also \ord$(17,16)=1$ and so the Fourier $[16\ti
      16]$ matrix may be constructed over $GF(17)$; in this case the
      arithmetic is modular arithmetic in $\Z_{17}= GF(17)$ which is
      good. Primitive $16^{th}$ roots of unity here are $3 \mod 17$ and
      $5 \mod 17$; there are others.    

\item Suppose a rate $\frac{2}{5}$ is required which can correct $5$
      errors. Thus a code $[[n,k,d]]$ is required where $\frac{k}{n}=
      \frac{2}{5}=R$ and $d\geq 11$. Then $n=\frac{2(d-1)}{1-R}$. Now
      $1-R = \frac{3}{5}$ and it is required that $3/(2(d-1))$ for $d\geq
      11$. Thus take $d=13$ and get $n=40, k=16, r=28$. Let $F_{40}$ be
      a Fourier matrix of size $40\ti 40$. Take $r=28$ rows to form a
      $[40,28,13]$ dual-containing code from which by CSS construction a
      $[[40,16,13]]$ quantum mds code is obtained. Since $41$ is prime the Fourier $40\ti 40$ matrix may be taken over $GF(41) = \Z_{41}$ and a primitive $40^{th}$ root of unity is $7 \mod 41$. 
 
\item Suppose it is required to construct a rate $\frac{7}{8}$ quantum
      code which can correct $3$ errors, that is, construct a
      $[[n,k,d]]$ quantum code with $\frac{k}{n}= \frac{7}{8}, d\geq
      7$. Then $n=\frac{2(d-1)}{1-R})$ from which 

% Let $2r=n+k$. Now construct a $[n, r, d]$ dual-containing code. Require
% $7=d= n-r+1$. Thus require $7= n-\frac{n+k}{2} + 1$ and so
% $n-\frac{n+\frac{7n}{8}}{2}=6$. From this 
 $n=96, k=84, r= 90$ and $d=7$ works. Thus
construct a $[96,90,7]$ dual-containing code from which the $[[96,
84,7]]$ quantum code is constructed by the CSS construction. 

To construct a dual-containing $[96,90,7]$ code, a Fourier $96 \ti 96$
matrix is required. Now $5$ is the first prime $p$ with $\gcd(p,96)=1$. 
Then $\phi(96) = 32$ and so the field $GF(5^{32})$ will suffice. However
the order of $5 \mod 96$ is $8$ and so the field $GF(5^8)$
works. Construct the $96\ti 96$ Fourier matrix over $GF(5^8)$ and take
the code generated by the first $90$ rows of this matrix to get a
$[96,90,7]$ dual-containing code. Also \ord$(97,96)=1$ and hence the
      arithmetic could be done over $GF(97)=\Z_{97}$.   

\item It is required to construct a rate $\frac{15}{16}$ quantum code
      which can correct $3$ errors. Thus a $[[n,k,\geq 7]]$ quantum
      code with $\frac{k}{n} = \frac{15}{16}$ is required. Hence $k=
      \frac{15n}{16}$. Let $2r= n+k$ and then require an $[n,r,d]=[n,r,\geq 7]$
      dual-containing code. Then from above get $n=\frac{2(d-1)}{1-R}= 32(d-1)$. For $d=7$ this gives  $n=192, k=180,r=186$. Thus it is required to construct a Fourier
      $192 \ti 192$ matrix and take the first $186$ rows to get a $[192,
      186, 7]$ dual-containing code. Using the CSS construction a
      $[[192, 180,7]]$ quantum code is obtained. 

Over which fields can a Fourier $192\ti 192$ be constructed? Now $193$
      is prime itself and let $\F= GF(193)=\Z_{193}$ will work.  An
      element of order $192$ in $\F$ is required. Now \ord$(5,193)=192$
      and so $5 \mod 192$ is a primitive element from which the $192\ti
      192$ Fourier matrix over $GF(193)$ can be constructed. The
      arithmetic is modular arithmetic. There are fields with smaller
      characteristic which work but their order is larger. For example
      $3^{16} \equiv 1 \mod 193 $ and so there exists a primitive $192$
      root of unity in $GF(3^{16})$. 

Suppose we require a rate $\frac{15}{16}$ quantum code which now can
      correct $7$ errors. Hence require a $[[n,k,\geq 15]]$ quantum code. The
      calculations are similar and it is found that $n=448, k= 420, r=
      434$ gives a $[448, 434,15]$ dual containing code from which a
      $[[448, 420, 15]]$ quantum code may be constructed. It is required
      to find a Fourier $448\ti 448$ matrix. Now $449$ is prime and so
      such a matrix may be found over $GF(449)=\Z_{449}$. Here $3 \mod
      449  $ has order $448$ and thus may be used as the primitive
      element to form the $448\ti 448$ Fourier matrix over
      $GF(449)=\Z_{449}$.
  
\end{itemize}

\subsection{Given field, find best quantum code} Let $\F=GF(q)$ be a given finite field. What is the maximum length of a quantum code that could be constructed with coefficients from $\F$? Let $\om$ be a primitive element in $\F$; thus $\om^{q-1} = 1, \om^r \neq 1 , 1\leq r< q-1$.
  
Now construct the Fourier  $(q-1)\ti (q-1)$ over $\F$ using $\om$ as the primitive root of unity. Then dual-containing $[q-1,r, d]$ mds codes may be constructed over $\F$ provided $r> n/2$ and then quantum $[[q-1, 2r-(q-1), d]]$ codes may be constructed by the CSS construction. 
\subsection{Required over a field of prime order}
Given a rate and a required error-correcting capability,  can a quantum code with this rate and this error-correcting  be constructed over a prime field? In this case the arithmetic is then modular arithmetic which is easy to implement. 

%% Here is an example which shows how this can be done. Suppose a quantum
%% code $[[n,r,d]]$ of rate $\frac{3}{5}=R$ and distance $11$ is
%% required. By Section \ref{required} this requires $n =
%% \frac{2(d-1)}{(1-R)} = 50$ in order to construct a $[50,40, 11]$
%% dual-containing code from which by CSS construction a $[[50,30,11]]$
%% quantum code is deduced. The $[50,40,11]$ is formed from a Fourier $ 50
%% \ti 50$ matrix in a field which has a $50^{th}$ root of unity. Now $51$
%% is not prime and the nearest prime greater than $51$ is $53$. So
%% consider the field $GF(53) = \Z_{53}$ which has a primitive $52^{nd}$
%% root of unity. Form the Fourier $52 \ti 52$ matrix over $GF(53)$ and
%% take $42$ rows in succession, or in arithmetic sequence with difference
%% $k$ satisfying $\gcd(k,52) =1$, to form a $[52,42, 11]$ dual-containing
%% code. From this derive the $[[52,32,11]]$ quantum code. The rate of this
%% quantum code is $\frac{32}{52}$ is slightly larger than required rate
%% $\frac{3}{5}$ but this is okay for our purposes and the distance is as
%% required exactly.  Now \ord$(2,53)$ is $52$ so $ 2 \mod 53$ could be
%% taken as the primitive $52^{nd}$ root with which to form the Fourier
%% $52\ti 52$ matrrix. The arithmetic is $\mod 53$.  

%% In this example some compromise in the rate is required in order to get the exact distance. Suppose however the system can only tolerate this exact rate. 
Here is an example to show how this can be done. Suppose  a rate of $\frac{3}{5}=R$ and a distance $d\geq 11$ is required. Then from above calculations $n=\frac{2(d-1)}{1-R}$. Now $1-R=\frac{2}{5}$ and it is then required that $n=5(d-1)$. Now look at the distances of $d=11, d=12, d=13, ...$ and then require $n=50, n= 55,n=60, ...$ respectively. For $n=60$ it is seen that $61$ is prime. Consider then $GF(61)$ and form the $60\ti 60$ Fourier matrix $F_{60}$ using a primitive $60^{th}$ root of unity in $GF(61)$. Now \ord$(2,61) = 60$ so $2 \mod 61$ could be used as the primitive $60^{th}$ root of unity. 

Then form the $[60, 48,13]$ dual-containing code from $F_{60}$ from which the $[[60, 36,13]]$ quantum code is deduced. This has rate $\frac{3}{5}$ as required. The distance is $13$  and  the arithmetic is modular arithmetic in $GF(61) =\Z_{61}$.  

%In general proceed as follows. %% If the $n$ needed is such that $n-1$ is not prime and if a rate slightly greater than the required rate can be tolerated then proceed as follows. %% It is not p
%% ossible to get the exact rate but by taking the `next' prime from previous constructions it is possible to get a rate (slightly)  greater than the required rate and the exact required distance is reached. 

Suppose then in general a  rate of $R$ and a distance $\geq d$ are
required for a $[[n,2r-n,t\geq d]]$ quantum code obtained from a
dual-containing $[n,r,t\geq d]$ over a prime field. By calculations in
Section \ref{required} it is required that $n=\frac{2(t-1)}{1-R}$ for a
distance $t$. Now for $1-R=\frac{s}{q}$ it is required further that
$s/(2(t-1))$. Look at $t=d, t=d+1, \ldots $ until $s/(2(t-1))$ and
$n+1=p$ is prime. Then $n$ determines $r$ and $2r-n$.  Then form the
$n\ti n$ Fourier matrix over  $GF(p) = \Z_p$ and form the
dual-containing $[n,r,t]$ code by taking $r$ (suitable) rows of this
Fourier matrix from which the $[[n,2r-n,t]]$ quantum code is
obtained. %The value of  %% and then the $[n,n+1-d,d]$ dual-containing
	  %code may be deduced from which the $[[n,n-2d+2,d]]$ quantum
	  %code is formed. Let now $p$ be the first prime $\geq
	  %n+1$. Form the Fourier $(p-1)\ti (p-1)$ matrix over $GF(p)$
	  %and from this derive the $[p-1,p-d , d]$ dual-containing
	  %code. Use the CSS construction to form the $[[p-1, p-2d+1,
	  %d]]$ quantum code. The rate of this code is
	  %$\frac{p-2d+1}{p-1} = \frac{p-2d+2-1}{p-1} = 1+
	  %\frac{2-2d}{p-1} \geq 1 + \frac{2-2d}{n} =1 -(1- R)=R$. The
	  %inequality in this is true  since $p-1\geq n$ and $d$ is
	  %positive.  

Can a prime number for $n+1$ always be obtained in this manner? In the cases looked at, a prime $n+1$ value is attained very quickly. It's not proved here in
general and is a conjecture. The following is a reasonably large example
that can be written down as an illustration.  

% If only the rate $R$ can be tolerated proceed as in the example. First get $n$ from $n=\frac{2(d-1)}{1-R}$ for required $d$. Then look at $d+1, d+2,...,$ and required $n$ until this $n$ satisfies $n+1=p$ is prime. Then work in prime field $GF(p)$ to get the required rate with larger difference but work can be done in $GF(p) =\Z_p$. 

Suppose for example a rate $\frac{4}{7}=R $ and distance $d\geq
17$, which can correct $\geq 8$ errors, are required. Then
$n=\frac{2(d-1)}{1-R} = \frac{14(d-1)}{3}$. Now it is necessary that $3/(d-1)$ and allowable values for $d\geq 17$ are
$d=19, 22, 25, \ldots$ giving $n+1=85,99, 113, \ldots$. Now $113$ is
prime so take $n=112 $ and then $r=88,k=64$. Then form a Fourier $112\ti
112$ matrix over $GF(113)=\Z_{113}$ and take suitable $88$ rows to form
the dual-containing $[112,88,25]$ code from which by the CSS
construction the $[[112,64,25]]$ mds quantum code which has rate
$\frac{4}{7}$. The arithmetic is then modular arithmetic in $GF(113)=\Z_{113}$. 
   Now \ord$(3,113)=112$ so $3\mod 113$ may be taken as the primitive $112^{nd}$ root of unity in forming the Fourier $112\ti 112$ matrix. 

\section{Hermitian}\label{her} Here codes from \cite{hurley} which are shown to be 
dual-containing relative to a Hermitian inner product are used to
construct mds quantum codes by the CSS construction. There is some
restriction on the rates achievable by the method % \footnote{Essentially
% rates less than $\frac{3}{4}$ are not achievable. }
 but then again using the
Hermitian inner product to obtain quantum codes may be useful. % on occasions. 

In a field $\F = GF(l^2)$ the Hermitian product of two vectors $u = (u_1,u_2, \ldots, u_n), v=(v_1,v_2, \ldots, v_n) \in \F^n$ is given by $<u,v>_H = u_1v_1^l + u_2v_2^l+ \ldots + u_nv_n^l$ and this is an element of $\F$.

Let $e_0,e_1, \ldots, e_{n-1}$ be the rows of a Fourier $n\ti n$ matrix. 
Now $<e_i,e_j>_H = <e_i,e_j^l>_E = <e_i,e_{jl}>$. Thus $<e_i,e_j>_H = 0$
except when $j*l=n-i$, that is when $i+j*l \equiv 0 \mod n$; in this
case $<e_i,e_j>_H =n$. %% When $i+j*l=\equiv 0 \mod n$ refer to $e_j$ as
		     %% the non-hermitian dual (NHD) of $e_i$.  

Let   $\underline{v} = (v_1,v_2, \ldots, v_n) \in \F^n$ and define 
$\underline{v}^l = (v_1^l,v_2^l, \ldots, v_n^l)$.

The following lemma is straightforward. 
\begin{lemma} For rows $e_i$ of the Fourier matrix, $e_i^l = e_{il}$ where $il$ means $i*l \mod n$. 
\end{lemma} 
  
Let $e_i$ be such a row of a Fourier matrix over a field $GF(l^{2})$. The Hermitian inner product  is 
$<v.u>_H = <v,u^{l}>_E$. Now for a row $e_i$  of the Fourier matrix,
$e_j$ is orthogonal to every $e_j$ except just one that satisfies
$<e_i,e_j>_H \neq 0$. % Then $<e_i,e_j>_H \neq 0$ if and only if
% $<e_,e_j^{l} \neq 0$ if and only if $ j$.

Now $<e_i,e_j>_H= <e_i,e_j^l> = (1,\om^i,\om^{2i}, \ldots, \om^{(n-1)i}). 
(1,\om^{jl},\om^{2(jl)}, \ldots, \om^{(n-1)jl})\T = 1 
+w^{i+jl}+w^{2(i+jl)}+\ldots + w^{(n-1)(i+jl)}$.

This sum is zero except when $i+jl \equiv 0 \mod n$. As $l$ has an
inverse $\mod n$, there is just one $j$ for each $i$. The dual of $e_i$ is generated by all the other $e_j$ except for
this  $e_j$ obtained from $i+jl \equiv 0 \mod n$.  
Say a row $e_i$ is self-dual relative to the Hermitian inner product if
$<e_i,e_i>_H = 0$ and otherwise say the row $e_i$ is non-self-dual
(relative to the Hermitian inner product). Thus $e_i$ is non-self-dual
(relative to Hermitian inner product) if $<e_i,e_i>_H \neq 0$. 
%For example: 

\subsection{Characteristic 2}
The characteristic 2 finite fields are  $GF(2^m)$. Interest now is in 
$GF(2^{2n})$ on which a Hermitian inner product  is defined
by $<v,u>_H = v_1u_1^{2^n}+ v_2u_2^{2^n} + \ldots + v_nu_n^{2^n}$. This
is $<v,u>_H = <v,u^{2^n}>_E$ where the suffix $E$ indicates the Euclidean
norm. The `$l$' here of the general case is $2^n$. 

\subsubsection{$GF(2^4)$}
Consider then $GF(2^4)$ initially. Now $2^4-1 = 15$ and there exists a
primitive $15^{th}$ root of unity in $GF(2^4)$. Thus form the Fourier
$15 \ti 15$ matrix  over $GF(2^4)$ with a primitive $15^{th}$ root of unity. Any $r$ consecutive rows of this
matrix generates an $[15,r,15-r+1]$ mds code. Interest here is in when such
codes are dual-containing relative to the Hermitian inner product. Here $l=2^2=4$.

The rows of the Fourier $15\ti 15 $ matrix using the primitive
$15^{th}$ root of unity in $GF(2^4)$ are $\{e_0,e_1, \ldots, e_{14}\}$. The non-self-dual rows are 
$\{e_0, e_3, e_6$, ..., $e_{12}\}$. %% are the non-self-dual rows 
%% and all the others are self-dual.
Thus take $\mathcal{C} = <e_0,e_1,\ldots, e_{12}>$. Then
$\mathcal{C}^{\perp_H} $ consists of all the $e_i$ which are *not* in
the set $\{e_0, e_{11},e_{7},e_3,e_{14}, e_{10}, e_6, e_2, e_{13},
e_9,e_5, e_1,e_{12}\}$.  This is $\{e_4,e_8\}$ and constitutes 
$\mathcal{C}^{\perp_H}$. Thus $\mathcal{C}$ is dual-containing and is a
$[15,13,3]$ mds code by Theorem \ref{seq1}.

By CSS construction this gives a $[[15,11,\geq 3]]$ quantum mds code. For the
general case later on,  note that this is a 
$[[2^4-1, 11, 2^2-1]]$ mds quantum code. 

The question occurs as to whether it is possible to get other dual-containing codes from $F_{15}$. Take $\mathcal{D} = <e_0,e_1,e_2,
 \ldots , e_9, e_{12}>$; that is, skip from 9 to 12. It is not possible
 to leave any of $\{e_0,e_3,e_6,e_9, e_{12}\}$ out of a %possible
 dual-containing  code as these are non-self-dual elements. % as then this one would appear in the dual and so the code 
%  would not be dual containing. 
 
Now $\mathcal{D}^{\perp_H} = <e_1,e_4,e_5,e_8>$ and so the code is dual-containing. Hence $\mathcal{D}$ is a dual-containing $[15,11]$ code but we
cannot use Theorem \ref{seq} or Theorem \ref{seq1} to find its distance as the rows are not
consecutive  nor in arithmetic sequence with difference $k$ satisfying
$\gcd(k,15)=1$. What is the distance of $\mathcal{D}$?

Let $\mathcal{F} = < e_0, e_2, \ldots, e_6, e_8,e_9,e_{12}>$. This is
also (Hermitian) dual-containing and its dual is $\mathcal{F}^{\perp_H}
= <e_1,e_2,e_4,e_5,e_8>$. This is a $[15,10]$ dual-containing code but
we cannot immediately deduce its distance as Theorem \ref{seq} or
Theorem \ref{seq1} cannot be applied. 

However we may also produce different $[[15,11,3]]$ mds quantum codes as follows.
Choose rows $\{e_3, e_4, \ldots, e_{14}, e_0\}$ leaving out $e_2,e_3$ to form a dual-containing code. It is a $[15,13,3]]$ code by Theorem \ref{seq1} since the rows are consecutive. This then by CSS construction gives a $[[15,11,3]]$ mds quantum code. Choose rows $e_6, e_7, \ldots, e_{14}, e_0,e_1,e_2, e_3$ to give another $[15,13,3]$ dual containing code; the rows are consecutive so get the full distance. The $[[15,11,3]]$ mds quantum code is derived. In general take the rows $\{e_{3t}, e_{3t+1},\ldots, e_{0}, e_1, \ldots , e_{3t-3}\}$ to form a $[15,13,3]$ dual containing code from which to derive a $[[15,11,3]]$ mds quantum code. 

\subsubsection{$GF(2^6)$}
Consider $GF(2^6)$. Now $2^6-1= 63$ and let $F_{63}$ be the $63 \ti
63 $ Fourier matrix obtained from the primitive $63^{rd}$ root of unity in
$GF(2^6)$. Here our $l=2^3= 8$ for the Hermitian inner product. 

The rows of $F_{63}$ are $\{e_0, e_1, \ldots, e_{62}\}$. % These are the $i$ which are multiples of the solution to
% $63-x*l=x$ which is $x=6$.
 Then  $\{e_0,e_7,\ldots, e_{56}\}$ are the
non-self-dual elements. Hence take $\mathcal{C}= < e_0, e_1, \ldots, e_{56}>$. This
gives an $[63,57, 7]$ dual-containing mds code by Theorem \ref{seq1}.  By CSS
construction a $[[63,51, 7]]$ quantum mds code is obtained. Note this is of the form $[[2^6-1, 51, 2^3-1]]$. % mds quantum code.

Further $[63, 57,7]$ dual-containing codes are obtained by $\mathcal{D} = 
< e_6, e_7, \ldots, e_{62}, e_0 >$ or more generally by starting at one
of the non-self-dual $e_i$ and going by sequence to $e_{i-6}$.

%% Dual-containing codes of smaller dimension may also be obtained from $F_{63}$ 
%% but the distances cannot be calculated by Theorem \ref{seq} or Theorem 
%% \ref{seq1}.  

\subsubsection{$GF(2^8)$} %Consider one more case before stating the general result. 
%This case is of the length of the Reed-Solomon codes.
 
Consider $GF(2^8)$. Now  $2^8-1=255$ and there exists a primitive $255^{th}$ root
of unity in $GF(2^8)$. Form the Fourier $255\ti 255$ matrix $F_{255}$ with this
primitive root and denote the rows of $F_{255}$ by $\{e_0, e_1, \ldots,
e_{254}\}$. The $l$ here is $l=2^4=16$.

The non-self-dual $e_i$ are %% the multiples of $x$ satisfying $255-x*l=x$ and
%% thus $x=15$. Hence 
$\{e_0,e_{15}, \ldots, e_{240}\}$. %%  are the non-self-dual
%% $e_i$.
Then take $\mathcal{C} = <e_0,e_1, \ldots, e_{240}>$ to get a dual
containing $[255, 241, 15]$ mds code by Theorem \ref{seq1}. Use  CSS construction to get an
$[[255, 227,15]]$ quantum code. This has the form $[[2^8-1,227,2^4-1]]$. 

%% Since the corresponding linear codes have efficient decoding
%% algorithms by \cite{hurley}  these quantum codes also have efficient
%% decoding algorithms.       

\subsection{General characteristic 2}
General method: Consider $GF(2^{2n})$ and $q = 2^{2n} -1, l=2^n$. Form
the Fourier $q\ti q$ matrix over $GF(2^{2n})$ using a primitive $q^{th}$ root of unity. Now  $\{e_0, e_{2^n-1}, \ldots,  e_{2^n(2^{n}-1)}\}$ are the non-self-dual elements of the Fourier matrix.  

Then let $\mathcal{C}= <e_0, e_1, \ldots, e_{2^n(2^n-1)}>$ which is then a dual containing $[2^{2n}-1,
2^{2n}-2^n+1, 2^n-1]$ mds code by Theorem \ref{seq1}. Use the CSS construction to form the $[[2^{2n}-1,2^{2n}-2^{n+1}+3, 2^n-1]]$ quantum mds code. %% From $2d\leq n-k+2$ get $2d \leq 2^{2n} -1 -2^{2n} +2^{n+1} -3 +2 = 2^{n+1}-2$ and so $d\leq 2^{n} -1 $ giving that $ d= 2^n-1$. Thus  
%% an $[[2^{2n}-1,2^{2n} -2^{n+1} + 3 , 2^n-1]]$ mds quantum code is
%% obtained in this manner.

The rate of such a code is $R_n= \frac{2^{2n} -2^{n+1} + 3}{2^{2n}-1}$.
Now $\lim_{n\rightarrow \infty}R_n=1$.  %% Rates less than $11/15$ are not
%% achievable but by taking $n$ large enough any distance is achievable.

\subsection{Characteristic 3} The finite fields  of characteristic 3 are $GF(3^t)$. Interest here is in $GF(3^{2n})$. The `$l$' here is $3^n$. 

Consider first $GF(3^2)$. Here $3^2-1=8$ and so $GF(3^2)$ contains a primitive $8^{th} $ root of unity. Form the Fourier $8\ti 8$ matrix with a primitive $8^{th}$ root of unity. Denote its row by $\{e_0,e_1, \ldots, e_7\}$. Here $l=3$. 
 % The $i$ of the non-self-dual $e_i$ are obtained from multiples of the solution $x$ of $8-x*l= x$ giving $x=2$.
 The non-self-dual $e_i$ are $\{e_0, e_2, e_4, e_6\}$.
Consider then $\mathcal{C} = <e_0, e_1, \ldots, e_6>$ which is a dual-containing $[8,7,2]$ mds linear code. By the CSS construction this gives the $[[8, 6, 2]]$ quantum code which is mds. This is not very exciting.

Now look at $GF(3^4)$. Then $3^4-1 = 80$ and thus $GF(3^4)$ contains a primitive $80^{th}$ root of unity. Form the Fourier $80\ti 80$ matrix with one of these primitive roots. The rows are $\{e_0, e_1, \ldots, e_{79}\}$. The non-self-dual $e_i$ are %% are determined from the solution of $80-x*l=x$. Here $l=3^2=9$ and so $x=8$. Thus the non-self-dual $e_i$ are
$\{e_0,e_8, e_{16}, \ldots, e_{72}\}$.
Form $\mathcal{C} = <e_0, e_1, \ldots, e_{72}> $ to get a $[80,73, 8]$ mds dual containing code. This then gives by the CSS (Hermitian) construction the $[[80, 66, 8]]$ quantum mds code. % which must be a $[[80,66,8]]$ mds quantum code. 

$GF(3^6)$ has $3^6-1= 728$ thus giving a $728$ primitive root of unity in $GF(3^6)$. Form the Fourier $728 \ti 728$ with this root of unity and denote its rows by $\{e_0, e_1, \ldots, e_{727}\}$. The non-self-dual $e_i$ %% are determined by multiples of the solution of $728-l*x=x$ where $l=3^3=27$. This gives $x=26$ and the non-self-dual $e_i$
are $\{e_0,e_{26}, e_{52}, \ldots, e_{702}\}$. Form $\mathcal{C} = <e_0,e_1, \ldots, e_{702}> $ to get the dual-containing $[728, 703,26]$ mds code. Using the CSS construction gives a $[[728, 678, 26]]$ mds quantum code. 

%% One step further gives the $[[3^8-1,3^8+1-2(3^4)+3, 3^4-1]]$ mds quantum code. %
%What is $r$? 
% Need $2(3^4-1) = 3^8-1-r+2= 3^8+1-r $ and so $$. 

In general by considering $GF(3^{2n})$ an $[[3^{2n}-1,
3^{2n}-2.3^n+3, 3^n-1]]$ quantum mds code is obtained. % The $r$ is
% defined in terms of the length and distance as the code is mds so (using
% $2d \leq n-r+2$ for a quantum code:) $2(3^n-1)=3^{2n}-1 -r +2$ gives the
% value for $r$ which would in any case come from the construction.
The rate is $R_n= \frac{3^{2n}-2.3^n+3}{3^{2n}-1}$ and
$\lim_{n\rightarrow \infty}R_n=1$. %% Rates less than $\frac{3}{4}$ are not
%% achievable but by taking $n$ large enough any distance is achievable. 
   
\subsection{Characteristic 5} 
A detailed working out for $GF(5^4)$ is given before the general statement. 
Now $5^4-1 = 624$ and so define the Fourier $624 \ti 624$ matrix $F_{624}$ over $GF(5^4)$ using a primitive $624$ root of unity. Denote the rows of $F_{624}$ by $e_0,e_1, \ldots, e_{623}$. Here $l=5^2$ for the Hermitian inner product. 
The non-self-dual $e_i$ are %% multiples of the solution $x$ of $624 - x*l= x$ giving that $x=24$. Thus the self dual are
$e_0,e_{24}, e_{48}, \ldots, e_{600}$.
Consider $\mathcal{C} = <e_0,e_1, \ldots, e_{600}>$  which is then a $[624, 601, 24]$ dual-containing mds code by Theorem \ref{seq1}. The CSS construction then gives a $[[624, 578, 24]]$ quantum mds code. %% which is mds and so is a $[[624,578, 24]]$ quantum code which can correct $11$ errors.

In general for characteristic $5$, $[[5^{2n}-1,5^{2n}-2.5^n+3,
5^{n}-1]]$ quantum mds codes may be  obtained from Hermitian CSS construction. The rate is
$R_n=\frac{5^{2n}-2.5^n+3}{5^{2n}-1}$ and $\lim_{n\rightarrow
\infty}R_n=1$. % . The $r$ is such that the code is mds but comes from the general construction. 
\subsection{General characteristic} Suppose the characteristic is $p \neq 0$. Then quantum mds codes of the form 
$[[p^{2n}-1, p^{2n}-2.p^n+3, p^{n}-1]]$ may be formed 
using the Hermitian version of the CSS construction. 
The details are omitted.
%The rate is $R_n= \frac{p^{2n}-2.p^n+3}{p^{2n}-1}$ and
%$\lim_{n\rightarrow \infty} R_n=1$. %% Not every small rate is achievable %% (for
%% %% $p\neq 2$ the rates not achievable may be found from $n=2$)
%% but any
%% distance is achievable by letting $n$ be big enough. 
%% % The $r$ is such that the code
%is mds but comes from the general construction. 

%It is noted in \cite{mackay} that for quantum codes 
%it is desirable to have dual-containing codes where the low length
%codewords occur in the dual or else the low length codewords not in
%the dual do not adversely affect the quantum code. 
%Here we have the
%situation that the codeword $c$ is short if and only if the input word
%$w$ is short. A suggestion then  would be to  {\em avoid short words}
%or perhaps to
%take the conjugate of short words and mark these. 

%%  \cite{brower}  contains tables  of some known
%% quantum codes with best known distances up to length $128$.  

%% The binary self-dual and dual containing codes
%%   as described in \sref{construct} and \sref{dualcontain1} 
%% may   be considered as codes over
%%   $GF(4)$ with the same length and distances.% These are self-dual or
  %dual-containing 
  %under the Euclidean inner product. 

%% The primitive element is
%%   not involved in these constructions. Using the primitive element and
%%   the  symplectic inner product, further classes  of self-dual
%%   and dual-containing codes over $GF(4)$ are obtained as follows:

%% \noindent National University of Ireland, Galway\\Galway\\
%% Ireland. \\email: ted.hurley@nuigalway.ie
\end{document}